\documentclass[12pt]{article}
\usepackage{fullpage}
\usepackage{amsmath}
\usepackage{amssymb}
\usepackage{amsthm}
\usepackage{subfigure}
\usepackage{tikz}
\usepackage{color}

\newtheorem{claim}{Claim}
\newtheorem{property}{Property}
\newtheorem{theorem}{Theorem}
\newtheorem{lemma}{Lemma}

\newcommand{\ignore}[1]{}
\newcommand{\changed}[1]{{#1}}
\newcommand{\changedA}[1]{{#1}}
\newcommand{\changedB}[1]{{#1}}

\newcommand{\changedD}[1]{{#1}}

\newcommand{\fix}[1]{{\color{red}#1}}

\begin{document}
\title{Flip Distance Between Two Triangulations of a Point-Set is NP-complete}
\date{}
\author{Anna Lubiw$^*$ \and Vinayak Pathak\thanks{Cheriton School of Computer Science, University of Waterloo, Waterloo, Canada\texttt{\{alubiw,vpathak\}@uwaterloo.ca}}}

\maketitle
\begin{abstract}
Given two triangulations of a convex polygon, computing the minimum number of flips required to transform one to the other is a long-standing open problem. It is not known whether the problem is in P or NP-complete. We prove that two natural generalizations of the problem are NP-complete, namely computing the minimum number of flips between two triangulations of (1) a polygon with holes; (2) a set of points in the plane.
\end{abstract}

\section{Introduction}
\changedA{
Given a triangulation in the plane, a \emph{flip} operates on two triangles that share an edge and form a convex quadrilateral.  The flip replaces the diagonal of the convex quadrilateral by the other diagonal to form two new triangles. 
A sequence of flips can transform any triangulation to any other triangulation---this is true for triangulations of a convex polygon, and more generally for triangulations of  a polygonal region with holes, which includes the case of triangulations of a point set.

In this paper we investigate the complexity of computing the \emph{flip distance} which is the minimum number of flips to transform one triangulation to another.
This is particularly interesting for convex polygons, where the flip distance is the rotation distance 
between two binary trees (see below).

The main result of our paper is that it is NP-complete to compute the flip distance between two triangulations of a 
a polygon with holes, or a set of points in the plane.}

\subsection{Flip distance and rotation distance}
\label{sec:rotation}
Balanced binary search trees are a widely used data structure. One way to make a rooted binary search tree balanced is using an operation called a rotation~\cite{CLR01}. Despite being very simple and fundamental, the rotation operation is not completely understood. 
In particular, the complexity of the problem of computing the minimum number of rotations needed to convert one rooted and labelled binary search tree to another, 
called the rotation-distance, has been open since 1987~\cite{STT87, DHJLTZ97}. 
It is not known if the problem is NP-complete.
The problem is also closely related to Sleator and Tarjan's famous dynamic optimality conjecture~\cite{ST85}.

Rooted binary trees are of interest to researchers in the field of bioinformatics as well.  An evolutionary tree is a rooted binary tree and 
constructing an evolutionary tree that fits a given set of species given some information about their DNA sequences is a widely studied problem. 
Evaluating the effectiveness of methods of constructing evolutionary trees leads naturally to the problem of 
measuring how similar two trees are to each other. Several measures of similarity have been used including the rotation distance between the two trees~\cite{DHJLTZ97}.

There is a bijection between binary trees with $n-1$ labeled leaves and triangulations of an $n$-vertex convex polygon. Moreover, a rotation in the tree corresponds to a flip in the polygon.  Thus, computing the rotation distance between two trees is exactly equivalent to computing the flip distance between two triangulations of a convex polygon.  See~\cite{STT88}.
 
\ignore{
Consider a triangulated convex polygon and define a \emph{flip} operation as follows. Pick any two triangular faces $ABD$ and $CBD$ that share a common edge $BD$ such that $ABCD$ is a convex quadrilateral and replace edge $BC$ with edge $AD$. This operation is called a flip.}
\ignore{
It is well known (see~\cite{STT88}) that computing the rotation distance between two trees is exactly equivalent to computing the 
flip distance between two triangulations between a convex polygon.

The two problems stated above are equivalent because there is a bijection between binary trees with $n-1$ labeled leaves and triangulations of an $n$-vertex convex polygon. Moreover, the triangulations that are one flip apart correspond to binary trees that are one rotation apart.}

\subsection{Generalizations and related work}
\label{sec:previous-work}
Flip distance between triangulations of a convex polygon and rotation distance between binary trees have been well studied in the past. Several results deal with the combinatorics of the flip operation. Sleator et al.~\cite{STT88} proved that for large values of $n$, the flip distance between two triangulations of an $n$-gon is at most $2n-10$ and that occasionally $2n-10$ flips are necessary. Other results deal with the \emph{flip graph}, i.e., the graph where nodes correspond to triangulations of an $n$-gon and an edge between two nodes denotes the fact that the corresponding triangulations are one flip apart. For example, Lucas~\cite{Luc88} showed that the flip graph is hamiltonian.  See also Eppstein~\cite{Epp}.

Flips have been studied in more general settings as well. Dyn et al.~\cite{DGR93} proved that any two triangulations of a simple polygon can be transformed into one another using flips. They proved that the same holds even for two triangulations of a simple polygon with points inside it. \ignore{The latter is also known in the literature as a \emph{straight line embedding of a near-triangulation}, or simply a \emph{near-triangulation}. This is because the graph whose nodes correspond to the points and edges correspond to the edges of the triangulation and of the polygon is a maximal planar graph (also called a triangulation) where one particular face is not necessarily a triangle.} Lawson~\cite{Law77} proved an upper bound of $O(n^2)$ flips needed in any such flip sequence. Hurtado et al.~\cite{HNU99} proved that the bound is tight asymptotically.

Triangulations of polygons with interior points have been further generalized to triangulations of simple polygons with polygonal holes, also called polygonal regions. It is known that two triangulations of the same polygonal region can be transformed into each other using flips (see~\cite{OB08}). 
Note that a one-vertex polygonal hole is just a point. Thus triangulations of polygons with interior points are a special case.

Flips have also been studied in a more combinatorial setting. For example, given a maximal planar graph, we can define a flip as replacing an edge with another so that the resulting graph is also maximal planar. Wagner~\cite{Wag36} proved that given two maximal planar graphs $G_1$ and $G_2$, there always exists a sequence of edge flips that transforms $G_1$ into a graph isomorphic to $G_2$. Combinatorial bounds on the number of flips required have also been studied and the best known upper bound is by Bose et al.~\cite{BJRSV11} of $5.2n-24.4$. 

In the combinatorial setting, we have the choice of labelled vs.~unlabelled graphs. 
Sleator et al.~\cite{STT92} proved that $O(n\log n)$ flips are sufficient to transform one labelled maximal planar graph with $n$ vertices into another with the same vertices, and $\Omega(n\log n)$ flips are sometimes necessary.
Needless to say, this is a huge area with numerous directions of investigation. Bose and Hurtado~\cite{BH09} provide a survey.

Regarding the question of actually computing the flip distance, to the best of our knowledge, only triangulations of convex polygons have been studied and the question has been open since 1987~\cite{STT87, DHJLTZ97}. The best known result is a trivial factor-2 approximation algorithm, which can be improved under certain assumptions regarding the input~\cite{LZ98}. Recently it was proved that the problem is fixed parameter tractable in the flip distance~\cite{CJ09}. No hardness results are known either.

\ignore{
\section{Our contribution}
We propose that to tackle the long standing open problem about the complexity of flip distance between triangulations of a convex polygon, we first look at some more general problems. We show that computing the flip distance between two triangulations of a simple polygon with interior points is NP-complete. 

We show this in two steps. First we prove that the problem of computing the flip distance for triangulations of polygonal regions is NP-complete. Next we reduce it to the same problem on triangulations of simple polygons with interior points. (For definitions, please refer to Sections~\ref{sec:convex-polygon} and \ref{sec:previous-work}.)
}

\section{Triangulations of polygonal regions}

\begin{theorem}
\changed{The following problem is NP-complete: Given two triangulations of a polygonal region with holes and a number $k$, is the flip distance between the two triangulations at most $k$?}
\end{theorem}

\subsection{Proof idea}
\label{sec:proof-idea}
Note that the problem lies in NP.
We prove hardness by giving a polynomial time reduction from vertex cover on 3-connected cubic planar graphs, which is known to be NP-complete~\cite{BKM97, WAN79}.

The idea is to take a planar straight-line drawing of the graph and create a polygonal region by replacing each edge by a ``channel'' and each vertex by a  ``vertex gadget''.   We then construct  two triangulations of the polygonal region that differ on the channels, and show that a short flip sequence corresponds to a small vertex cover in the original graph.

We begin by describing channels and their triangulations, because this gives the intuition for the proof.  A \emph{channel} is a polygon that consists of two 7-vertex reflex chains joined by two \emph{end} edges, as shown in Figures~\ref{edge-gadget-init} and \ref{edge-gadget-final}.  Note that every vertex on the upper reflex chain sees every vertex on the lower reflex chain and vice versa.  We identify two triangulations of a channel: a \emph{left-inclined triangulation} as shown in Figure~\ref{edge-gadget-init}; and a \emph{right-inclined triangulation} as shown in Figure~\ref{edge-gadget-final}.

A channel is the special case $n=7$ of the polygons $H_n$ of Hurtado et al.~\cite{HNU99}.
They prove in Theorem 3.8 that the flip distance between the right-inclined and left-inclined triangulations of $H_n$ is $(n-1)^2$.  
We include a different  proof in order  to generalize:
 
\changedD{
\begin{property}
\label{prop:channel}
Transforming a left-inclined triangulation of a channel to a right-inclined triangulation takes at least 36 flips.
\end{property}
\begin{proof}
In any triangulation of a channel, each edge of the upper reflex chain is in a triangle whose apex lies on the bottom reflex chain.  This apex must move from lower right ($B_7$) to lower left ($B_1$), in order to transform the left-inclined triangulation to the right-inclined triangulation.  
Similarly, each edge of the lower reflex chain is in a triangle whose apex lies on the upper reflex chain, and must move from upper left to upper right.
However, one flip can only involve one edge of the upper chain and one edge of the lower chain (no other 4 vertices form a convex quadrilateral), and thus can only move one upper and one lower apex, and only by one vertex along the chain.
Twelve triangles times six apex moves per triangle divided by two apex moves per flip gives a lower bound of 36 flips. 
\end{proof} 
 }

We now show that the number of flips goes down if a channel has a {\it cap}, an extra vertex that is visible to all the channel vertices, as shown in Figure 1(c).

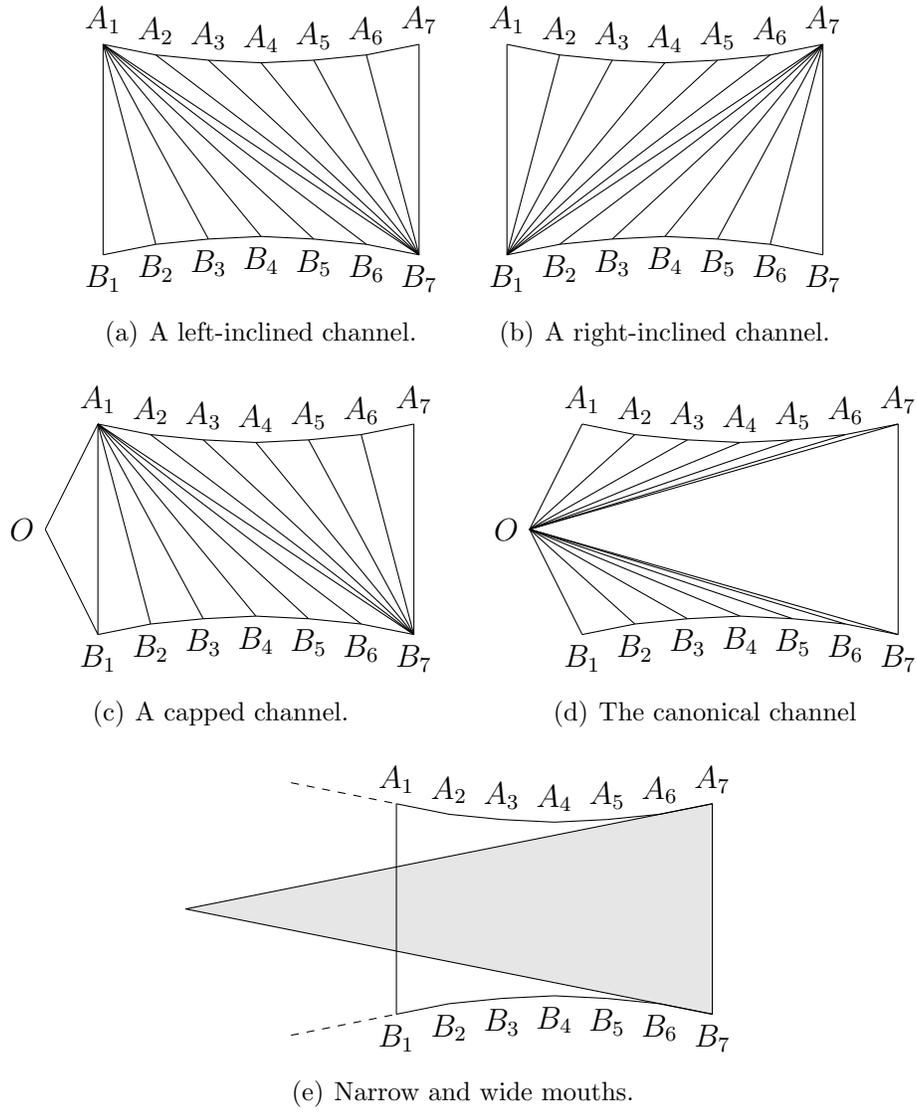
\begin{figure}
\centering
\subfigure[A left-inclined channel.]{
\label{edge-gadget-init}
\begin{tikzpicture}[scale=0.7]
\draw (-3,2) node[above]{$A_1$} --(-2,1.8); 
\draw (-2,1.8) node[above]{$A_2$} --(-1,1.7);
\draw (-1,1.7) node[above]{$A_3$} --(0,1.65);

\draw (0,1.65) node[above]{$A_4$} --(1,1.7);
\draw (1,1.7) node[above]{$A_5$} --(2,1.8);
\draw (2,1.8) node[above]{$A_6$} --(3,2) node[above]{$A_7$};

\draw (-3,-2) node[below]{$B_1$} --(-2,-1.8);
\draw (-2,-1.8) node[below]{$B_2$} --(-1,-1.7);
\draw (-1,-1.7) node[below]{$B_3$} --(0,-1.65);
\draw (0,-1.65) node[below]{$B_4$} --(1,-1.7);
\draw (1,-1.7) node[below]{$B_5$} --(2,-1.8);
\draw (2,-1.8) node[below]{$B_6$} --(3,-2) node[below]{$B_7$};

\draw (-3,2)--(-3,-2);
\draw (-3,2)--(-2,-1.8);
\draw (-3,2)--(-1,-1.7);
\draw (-3,2)--(0,-1.65);
\draw (-3,2)--(1,-1.7);
\draw (-3,2)--(2,-1.8);
\draw (-3,2)--(3,-2);

\draw (3,-2)--(-2,1.8);
\draw (3,-2)--(-1,1.7);
\draw (3,-2)--(0,1.65);
\draw (3,-2)--(1,1.7);
\draw (3,-2)--(2,1.8);
\draw (3,-2)--(3,2);

\end{tikzpicture}    
}                
\subfigure[A right-inclined channel.]{
\label{edge-gadget-final}
\begin{tikzpicture}[scale=0.7]
\draw (-3,2) node[above]{$A_1$} --(-2,1.8); 
\draw (-2,1.8) node[above]{$A_2$} --(-1,1.7);
\draw (-1,1.7) node[above]{$A_3$} --(0,1.65);
\draw (0,1.65) node[above]{$A_4$} --(1,1.7);
\draw (1,1.7) node[above]{$A_5$} --(2,1.8);
\draw (2,1.8) node[above]{$A_6$} --(3,2) node[above]{$A_7$};

\draw (-3,-2) node[below]{$B_1$} --(-2,-1.8);
\draw (-2,-1.8) node[below]{$B_2$} --(-1,-1.7);
\draw (-1,-1.7) node[below]{$B_3$} --(0,-1.65);
\draw (0,-1.65) node[below]{$B_4$} --(1,-1.7);
\draw (1,-1.7) node[below]{$B_5$} --(2,-1.8);
\draw (2,-1.8) node[below]{$B_6$} --(3,-2) node[below]{$B_7$};

\draw (-3,-2)--(-3,2);
\draw (-3,-2)--(-2,1.8);
\draw (-3,-2)--(-1,1.7);
\draw (-3,-2)--(0,1.65);
\draw (-3,-2)--(1,1.7);
\draw (-3,-2)--(2,1.8);
\draw (-3,-2)--(3,2);

\draw (3,2)--(-2,-1.8);
\draw (3,2)--(-1,-1.7);
\draw (3,2)--(0,-1.65);
\draw (3,2)--(1,-1.7);
\draw (3,2)--(2,-1.8);
\draw (3,2)--(3,-2);
\end{tikzpicture}
}
\subfigure[A capped channel.]{
\label{fedge-gadget-init}
\begin{tikzpicture}[scale=0.7]
\draw (-4, 0) node[left]{$O$} --(-3, 2);
\draw (-4, 0)--(-3, -2);

\draw (-3,2) node[above]{$A_1$} --(-2,1.8); 
\draw (-2,1.8) node[above]{$A_2$} --(-1,1.7);
\draw (-1,1.7) node[above]{$A_3$} --(0,1.65);
\draw (0,1.65) node[above]{$A_4$} --(1,1.7);
\draw (1,1.7) node[above]{$A_5$} --(2,1.8);
\draw (2,1.8) node[above]{$A_6$} --(3,2) node[above]{$A_7$};

\draw (-3,-2) node[below]{$B_1$} --(-2,-1.8);
\draw (-2,-1.8) node[below]{$B_2$} --(-1,-1.7);
\draw (-1,-1.7) node[below]{$B_3$} --(0,-1.65);
\draw (0,-1.65) node[below]{$B_4$} --(1,-1.7);
\draw (1,-1.7) node[below]{$B_5$} --(2,-1.8);
\draw (2,-1.8) node[below]{$B_6$} --(3,-2) node[below]{$B_7$};

\draw (-3,2)--(-3,-2);
\draw (-3,2)--(-2,-1.8);
\draw (-3,2)--(-1,-1.7);
\draw (-3,2)--(0,-1.65);
\draw (-3,2)--(1,-1.7);
\draw (-3,2)--(2,-1.8);
\draw (-3,2)--(3,-2);

\draw (3,-2)--(-2,1.8);
\draw (3,-2)--(-1,1.7);
\draw (3,-2)--(0,1.65);
\draw (3,-2)--(1,1.7);
\draw (3,-2)--(2,1.8);
\draw (3,-2)--(3,2);

\end{tikzpicture}    
}                
\ignore{\subfigure[A right-inclined capped channel.]{
\label{fedge-gadget-final}
\begin{tikzpicture}[scale=0.7]
\draw (-4, 0) node[left]{$O$} --(-3, 2);
\draw (-4, 0)--(-3, -2);

\draw (-3,2) node[above]{$A_1$} --(-2,1.8); 
\draw (-2,1.8) node[above]{$A_2$} --(-1,1.7);
\draw (-1,1.7) node[above]{$A_3$} --(0,1.65);
\draw (0,1.65) node[above]{$A_4$} --(1,1.7);
\draw (1,1.7) node[above]{$A_5$} --(2,1.8);
\draw (2,1.8) node[above]{$A_6$} --(3,2) node[above]{$A_7$};

\draw (-3,-2) node[below]{$B_1$} --(-2,-1.8);
\draw (-2,-1.8) node[below]{$B_2$} --(-1,-1.7);
\draw (-1,-1.7) node[below]{$B_3$} --(0,-1.65);
\draw (0,-1.65) node[below]{$B_4$} --(1,-1.7);
\draw (1,-1.7) node[below]{$B_5$} --(2,-1.8);
\draw (2,-1.8) node[below]{$B_6$} --(3,-2) node[below]{$B_7$};

\draw (-3,-2)--(-3,2);
\draw (-3,-2)--(-2,1.8);
\draw (-3,-2)--(-1,1.7);
\draw (-3,-2)--(0,1.65);
\draw (-3,-2)--(1,1.7);
\draw (-3,-2)--(2,1.8);
\draw (-3,-2)--(3,2);

\draw (3,2)--(-2,-1.8);
\draw (3,2)--(-1,-1.7);
\draw (3,2)--(0,-1.65);
\draw (3,2)--(1,-1.7);
\draw (3,2)--(2,-1.8);
\draw (3,2)--(3,-2);
\end{tikzpicture}
}
}
\subfigure[The canonical channel]{
\label{fedge-gadget-canonical}
\begin{tikzpicture}[scale=0.7]
\draw (-4, 0) node[left]{$O$} --(-3, 2);
\draw (-4, 0)--(-3, -2);

\draw (-3,2) node[above]{$A_1$} --(-2,1.8); 
\draw (-2,1.8) node[above]{$A_2$} --(-1,1.7);
\draw (-1,1.7) node[above]{$A_3$} --(0,1.65);
\draw (0,1.65) node[above]{$A_4$} --(1,1.7);
\draw (1,1.7) node[above]{$A_5$} --(2,1.8);
\draw (2,1.8) node[above]{$A_6$} --(3,2) node[above]{$A_7$};

\draw (-3,-2) node[below]{$B_1$} --(-2,-1.8);
\draw (-2,-1.8) node[below]{$B_2$} --(-1,-1.7);
\draw (-1,-1.7) node[below]{$B_3$} --(0,-1.65);
\draw (0,-1.65) node[below]{$B_4$} --(1,-1.7);
\draw (1,-1.7) node[below]{$B_5$} --(2,-1.8);
\draw (2,-1.8) node[below]{$B_6$} --(3,-2) node[below]{$B_7$};

\draw (-4, 0)--(-2, 1.8);
\draw (-4, 0)--(-1, 1.7);
\draw (-4, 0)--(0, 1.65);
\draw (-4, 0)--(1, 1.7);
\draw (-4, 0)--(2, 1.8);
\draw (-4, 0)--(3, 2);

\draw (-4, 0)--(-2, -1.8);
\draw (-4, 0)--(-1, -1.7);
\draw (-4, 0)--(0, -1.65);
\draw (-4, 0)--(1, -1.7);
\draw (-4, 0)--(2, -1.8);
\draw (-4, 0)--(3, -2);

\draw (3,2)--(3,-2);
\end{tikzpicture}
}
\subfigure[Narrow and wide mouths.]{
\label{channel-feasible-region}
\begin{tikzpicture}[scale=0.7]
\path[draw, fill=gray!20] (-7, 0) -- (3, 2) -- (3, -2) -- (-7, 0);

\draw (-3, 2) -- (-3, -2);

\draw (-3,2) node[above]{$A_1$} --(-2,1.8); 
\draw (-2,1.8) node[above]{$A_2$} --(-1,1.7);
\draw (-1,1.7) node[above]{$A_3$} --(0,1.65);
\draw (0,1.65) node[above]{$A_4$} --(1,1.7);
\draw (1,1.7) node[above]{$A_5$} --(2,1.8);
\draw (2,1.8) node[above]{$A_6$} --(3,2) node[above]{$A_7$};

\draw (-3,-2) node[below]{$B_1$} --(-2,-1.8);
\draw (-2,-1.8) node[below]{$B_2$} --(-1,-1.7);
\draw (-1,-1.7) node[below]{$B_3$} --(0,-1.65);
\draw (0,-1.65) node[below]{$B_4$} --(1,-1.7);
\draw (1,-1.7) node[below]{$B_5$} --(2,-1.8);
\draw (2,-1.8) node[below]{$B_6$} --(3,-2) node[below]{$B_7$};

\draw[dashed] (-5, 2.4) -- (-3,2);
\draw[dashed] (-5, -2.4) -- (-3, -2);
\end{tikzpicture}
}

\caption{Channels}
\label{edge-gadget}
\end{figure}

\begin{property}
The flip distance from a left-inclined to a  right-inclined triangulation of a capped channel is 24.
\label{prop:capped-channel}
\end{property}
\begin{proof}
The ``canonical'' triangulation shown in Figure~\ref{fedge-gadget-canonical} is 12 flips away from both the left-inclined and the right-inclined triangulations of a capped channel:
To flip the left-inclined triangulation to the canonical triangulation, flip edges $A_1B_1,\ldots , A_1B_7$ followed by edges $A_2B_7,\ldots , A_6B_7$ in that order. Similarly for the right-inclined triangulation. 

\changedD{
For the lower bound, we follow the same  idea as above.  In any triangulation, each edge of the upper [lower] reflex chain is in a triangle whose apex is either the cap or a vertex of the lower [upper] chain.  
There are only two kinds of flips: (1) a flip involving the cap vertex, an edge of one chain, and a vertex of the other chain; and (2) a flip involving one edge of each chain.   A flip of type (1) moves the apex of only one triangle, and moves the apex to or from the cap.  
If a triangle is altered by a flip of type (1) then at least two such flips are required, one to move the apex to the cap and one to move the apex from the cap. 
If a triangle is only altered by flips of type (2), then, as above, it costs 3 flips to get the apex to its destination.  Thus the 12 triangles require at least 24 flips.
}
\end{proof}

With the description of channels in place, we now elaborate on the idea of our reduction.  
We create a polygonal region by replacing each edge in the planar drawing by a channel, and each vertex by a vertex gadget.
We make two triangulations of the polygonal region.  In triangulation $T_1$ all edge channels are left-inclined and in $T_2$ all edge channels are right-inclined.  The triangulations are otherwise identical.
We design vertex gadgets so that
making a few flips in a vertex gadget creates a cap for a channel connected to it. Since transforming a channel from left-inclined to right-inclined is 
less costly if it is capped, the minimum flip sequence that transforms all the channels 
is obtained by
choosing the smallest set of vertices that covers all the edges and using them to cap all the channels. Thus, intuitively, a
minimum flip sequence 
corresponds to a
minimum vertex cover. 

\changed{One complication is that we cannot construct a vertex gadget for a \emph{sharp} vertex---a vertex of degree 3 where one of the three incident angles in the planar drawing is $ \ge \pi$. 
Therefore, we first show how to eliminate sharp vertices.}
\changedB{Let $G$ be our given 3-connected cubic planar graph. 
Using a result of Rote~\cite{Rote}, we can find, in polynomial time, 
a {\it strictly convex} drawing of $G$ on a polynomial-sized grid.  (In fact, Tutte's algorithm would also suffice for our purposes.)
 {\it Strictly convex} means that each face is a strictly convex polygon.  Thus the only sharp vertices of this drawing are the vertices of the outer face.  We replace each sharp vertex $v$ by a 3-vertex chain $v_1, v_2, v_3$ as shown in Figure~\ref{figure:sharp}.  We claim that $G$ has a vertex cover of size $\le k$ if and only if the modified graph has a vertex cover of size $\le k + t$, where $t$ is the number of vertices on the outer face of $G$.  This is because there are only two ways to cover the edges $(v_1, v_2)$ and $(v_2, v_3)$: $\{ v_1, v_3\}$, which corresponds to $v$ being in the vertex cover of $G$;  or $\{v_2\}$, which corresponds to $v$ not being in the vertex cover of $G$.}
 
 \begin{figure}[htb]
\centering
\includegraphics[scale=1]{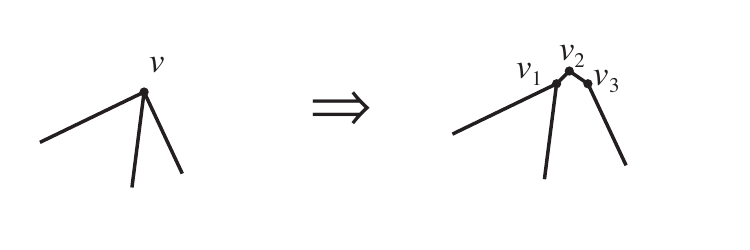}
\caption{Eliminating sharp vertices}
\label{figure:sharp}
\end{figure}



\subsection{Details of the reduction}
\label{sec:details}
For the remainder of the proof we will assume that we have a graph $G$ with vertices of degree 2 and 3, and a straight-line planar drawing $\Gamma$ of the graph on a polynomial sized grid with no sharp vertices.

For each channel, we define its \emph{narrow mouth} (the shaded region) and \emph{wide mouth} (the dotted lines) as shown in Figure~\ref{channel-feasible-region}.  Any point lying inside the narrow mouth and outside the channel can be a potential cap for the channel. 
\changedB{We argue below that a vertex outside the wide mouth does not help reduce the flip distance.} 

\changedA{
We now describe the triangulated vertex gadgets.  See Figures~\ref{degree-two} and~\ref{degree-three}.
Each of the 2 or 3 channels  attached to the vertex gadget will have one potential cap.
We place a convex quadrilateral $CDEF$ with diagonal $CE$, called the \emph{lock}, that separates each channel from its  potential cap.  Thus the lock $CE$ must be flipped, or ``unlocked'', in order to cap any channel.
}

\ignore{
Given the straight line embedding $P'$ of a graph $G'$ with no sharp vertices and maximum degree 3, we first create its two copies. Every construction that follows is repeated in the two copies thus giving us two triangulations $T_1$ and $T_2$ of a polygonal region. In the end, we make the channels in $T_1$ left-inclined and the channels in $T_2$ right-inclined.
\changed{Note that if a minimum vertex cover contains a degree-1 vertex, we can obtain another vertex cover of the same size by replacing the degree-1 vertex with its neighbor. Also, if a vertex cover does not contain any degree-1 vertices, then it must contain the neighbors of all degree-1 vertices. Thus, we begin our construction by removing all degree-1 vertices, their neighbors and the incident edges. The size of the minimum vertex cover of the resulting graph will remain the same.} 
\ignore{We begin by removing \fix{(we can't just remove them)} all degree-1 vertices because they can never be a part of VC$(G')$.}
Next, we place a channel aligned with each edge in the two copies. 
}

For the degree-2 gadget (see Figure~\ref{degree-two}), 
\changedB{
place point $C$ in the smaller angular sector (of angle $< \pi$) between the two channels, so that $C$ is outside the wide mouths of both channels.  Place points $D$, $E$, and $F$ in the other angular sector, with 
$D$ inside channel 1's narrow mouth and outside channel 2's wide mouth, 
$E$ outside the wide mouth of both channels, 
and 
$F$ inside channel 2's narrow mouth and outside channel 1's wide mouth.}
Triangulate as shown. \changedA{Thus $D$ is a potential cap for channel 1 and $F$ is a potential cap for channel 2.}
\ignore{
\fix{Should explain that $C$ is in the $< \pi$ sector}
Note that this construction is possible because the two edges make an angle at most $\pi -O(p(n))$ with each other.}

For the degree-3 gadget (see Figure~\ref{degree-three}), \changedB{note that because the vertex is not sharp, the mouth of each channel exits between the other two channels.  We place vertices in the angular sectors as shown in the figure.}
Place $D$ inside the intersection of the narrow mouths of channels 1 and 2,  and outside the wide mouth of channel 3. Place $F$ inside channel 3's narrow mouth and outside channel 1 and 2's wide mouths. Place $C$ and $E$ outside the wide mouths of all the channels and triangulate as shown. \changedA{Thus $D$ is a potential cap for both channel 1 and 2 and $F$ is a potential cap for channel 3.}
\ignore{ 
It is important for our construction that $DE$ crosses the narrow mouth of only channel 3, $EF$ crosses the narrow mouth of only channel 1 and $FC$ crosses the narrow mouth of only channel 2. This is possible because the vertex is not sharp.}

\changedB{Observe that every channel is blocked from its unique potential cap by exactly 3 edges.  (For example, in Figure~\ref{degree-three}, channel 1 is separated from its potential cap $D$ by edges $FA$, $FE$, and $CE$.)
Observe furthermore that  for each vertex gadget, the sets of blocking edges of the channels have one edge in common, namely the locking edge $CE$, and are otherwise disjoint.
These properties are crucial for correctness.}

We will say that a vertex gadget is \emph{locked} if the diagonal $CE$ exists and \emph{unlocked} otherwise. 
\changedD{
We first show what is possible with unlocked vertex gadgets.
}
\begin{property}
\label{prop:unlocked-channel-upper-bound}
If we unlock a vertex gadget then, for each channel attached to it, there is a sequence of 28 flips that transforms the channel triangulation and returns the vertex gadget to its (unlocked) state.
\end{property}
\begin{proof}
We first claim that there is a 2-flip sequence that caps the channel.   We enumerate the possibilities (refer to Figure~\ref{vertex-gadget}).
Note that we handle channels one at a time, not simultaneously.
 For the degree-2 gadget: flip $CF$ followed by \changedB{$CA$} for channel 1; flip $CD$ followed by \changedB{$CB'$} for channel 2. 
For the degree-3 gadget: flip $FE$ followed by \changedB{$FA$} for channel 1;  flip $CF$ followed by $CA'$ for channel 2; flip $ED$ followed by \changedB{$EA''$} for channel 3.
Once the channel is capped, we can transform the left-inclined triangulation to the right-inclined triangulation in 24 flips by Property~\ref{prop:capped-channel}.  Then we undo the 2 flips that capped the channel.  The total number of flips is 28. 
\end{proof}


\changedD{
Next we show some lower bounds on the number of flips.  First we note that the proof of Property~\ref{prop:channel} carries over to the following:

\begin{property}
\label{prop:locked-channel}
Transforming a left-inclined triangulation of a channel to a right-inclined triangulation takes at least 36 flips
even in the presence of other vertices, so long as 
the other vertices lie outside the wide mouths at either end of the channel.
\end{property}
}

\changedD{
We now consider what happens when  we unlock some vertex gadgets.  Let $T_1'$ be the triangulation obtained from $T_1$ by unlocking some vertex gadgets.  Let $T_2'$ be the triangulation obtained from $T_2$ by unlocking the same vertex gadgets.
Let $C$ be the set of channels that have a locked vertex gadget at both ends.   Then:

\begin{property}
\label{prop:unlocked-channel-lower-bound}
The number of flips required to transform $T_1'$ to $T_2'$ is at least $28|E-C| + 36|C|$.
%
\end{property}
\begin{proof}
Consider a  channel of $C$, with a locked vertex gadget at both ends.
The cap vertices of the channel are not useable.  By construction, the other vertices are outside the wide mouths of the channel.  Therefore, by Property~\ref{prop:locked-channel}, we need 36 flips to transform it.

Consider the  channels with an unlocked vertex gadget at one end.  
We only save flips by capping the channel.  To do this, we must flip the two edges that block the channel from its cap.  Because the edges that block one channel are disjoint from the edges that block any other channel, we must do two flips per channel, and we must re-flip those edges to return to the original state.
Finally, by Property~\ref{prop:capped-channel} it takes at least 24 flips to transform a capped channel.   (Note that the proof of  Property~\ref{prop:capped-channel} carries over even if the channel is capped at both ends.)  The total number of flips is 28 per channel.
\ignore{
Let $O$ and $O'$ be the two cap vertices. Let us assume that we cap the channel at at least one end and we flip at least one of $A_1B_1$ and $A_7B_7$ at some point, or otherwise the lower bound of 36 will apply. Capping and uncapping at one end, plus flipping one of $A_1B_1$ and $A_7B_7$ and flipping it back is 6 flips in total. We show that we must flip the remaining edges $A_1B_2,\ldots , A_1B_7$ and $B_7A_2, \ldots , B_7A_6$ (11 edges in total) at least twice, thus giving us a lower bound of $2\times 11+6 = 28$.

If an edge out of the 11 above is ever flipped to an edge that has either $O$ or $O'$ as its endpoint, then it will have to be flipped at least once more because no edges in the right-inclined triangulation end at $O$ or $O'$. An edge that never touches $O$ or $O'$ must still be flipped at least once because none of the 11 edges are common between the left-inclined and the right-inclined triangulations. Consider the first time it is flipped. The only edge that can be flipped for the first time without ending up touching $O$ or $O'$ is $A_1B_7$. After flipping, it becomes $A_2B_6$, which is not present in the right-inclined triangulation either and thus will have to be flipped again. Thus all 11 edges will have to be flipped at least twice.
}
\end{proof}
}

\ignore{
\begin{property}
\label{property-cap}
After unlocking a gadget from the state shown in Figure~\ref{vertex-gadget}, for each channel attached to it, there exists a 2-flip sequence that caps it.
\end{property}
\begin{proof} We enumerate the possibilities.  Note that we are not capping channels simultaneously.
For the degree-2 gadget: flip $CF$ followed by \changedB{$CA$} for channel 1; flip $CD$ followed by \changedB{$CB'$} for channel 2. 
For the degree-3 gadget: flip $FE$ followed by \changedB{$FA$} for channel 1;  flip $CF$ followed by $CA'$ for channel 2; flip $ED$ followed by \changedB{$EA''$} for channel 3.
\end{proof}}

\ignore{
\begin{property}
Given that we want to transform one or more of the channels incident to a gadget by capping them using the gadget while leaving the gadget in its original state in the end, the best way to do it is the following: first unlock the gadget, then for each channel to be transformed, cap it using its 2-flip sequence, transform it and then uncap by reversing the 2-flip sequence.
\end{property}
\begin{proof}
The only thing to check is that performing the 2-flip sequence for a channel does not help cap any other channel. This is trivial to see because the other channels still have two diagonals between them and their respective potential cap vertex.
\end{proof}}

\begin{figure}[htb]
\centering
\subfigure[Degree 2]{
\includegraphics[scale=1.3]{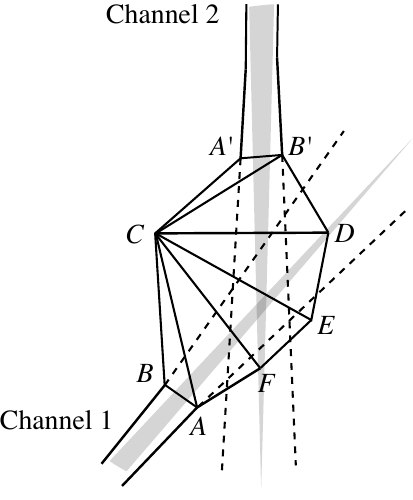}
\label{degree-two}
}
\subfigure[Degree 3]{
\includegraphics[scale=1.3]{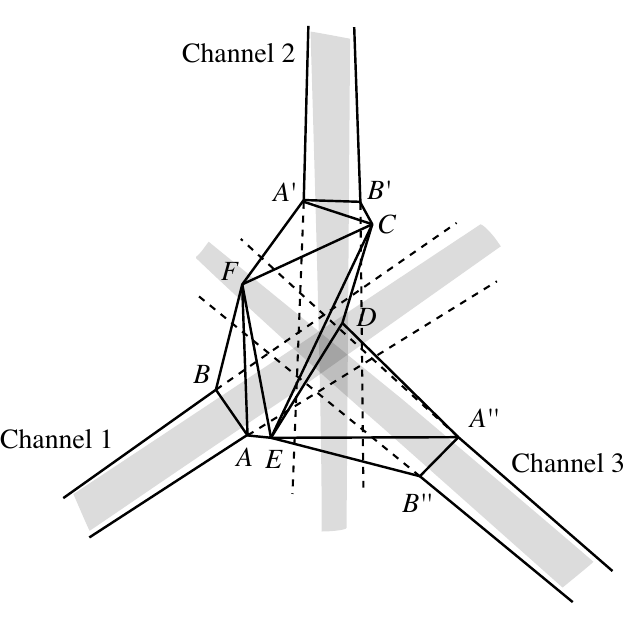}
\label{degree-three}
}

\caption{Gadgets for vertices}
\label{vertex-gadget}
\end{figure}

\subsection{Putting it all together}

\changedB{
\begin{lemma}
\label{main-lemma}
$G$ has a vertex cover of size $\le k$ if and only if the flip distance between the two triangulations $T_1$ and $T_2$ is $\le 2k + 28 |E|$.
\end{lemma}
\begin{proof}
Suppose that $G$ has a vertex cover of size $k$.  
Unlock the corresponding $k$ vertex gadgets.
Each edge channel has an unlocked gadget at one end, so by Property~\ref{prop:unlocked-channel-upper-bound} we can transform between the two triangulations of the channel in 28 flips.
When all channels have been transformed, we relock the $k$ vertex gadgets.
 The total number of flips is $2k + 28|E|$.

For the other direction, suppose that there is a flip sequence between $T_1$ and $T_2$ of length  $\le  2k + 28 |E|$.   
Let $L$ be the set of vertices whose gadgets are unlocked in the flip sequence.  
Let $C$ be the set of edges not covered by vertex set $L$.  
By adding one vertex to cover each edge of $C$, we observe that 
$G$ has a vertex cover of size $|L| + |C|$.  Thus it suffices to show that $|L| + |C| \le k$.
By Property~\ref{prop:unlocked-channel-lower-bound} the number of flips is at least
$2|L| + 36|C| + 28(|E - C|) \ge  2|L| +  28|E| + 8|C|$.  By assumption, the number of flips was $\le 2k + 28 |E|$.  Therefore $2|L| + 8|C| \le 2k$, which implies that $|L| + |C| \le k$, as required.
\end{proof}
}

The last ingredient of the NP-completeness proof is to show that the reduction takes polynomial time.  We need the following claim.

\begin{claim}
\label{claim:poly-size}
The size of the coordinates used in the construction is bounded by a polynomial in $n$.
\end{claim}
\begin{proof}
We give the main idea here, with further details in the appendix.  
We begin with a straight line drawing of a graph on a polynomial size grid.  Expand the grid, and allocate a square region around each vertex for the vertex gadget.  Expand each edge to two parallel line segments.  These line segments will become the channel, but for now, the reflex vertices of the channel are all collinear, which means that the channel's wide mouth is equal to its narrow mouth. 
The points $C,D,E,F$ of the  vertex gadget go in \emph{feasible} regions defined by the wide and narrow mouths (e.g.~in the 3-channel gadget, point $D$ lies in the narrow mouth of channels 1 and 2, but outside the wide mouth of channel 3).  We make the channels narrow enough so that all the feasible regions intersect the region allocated to the gadget.   
We claim that we can choose the channel end points $A, B, A', B', A'', B''$ on the expanded grid so that the resulting channels satisfy this property.

Now we pick points $C,D, E, F$ inside the appropriate regions.  Because the boundaries of the feasible regions are determined by pairs of points on the expanded grid,  the new points can be chosen to have polynomial size (because solutions to linear programs have polynomial size as shown in Theorem 10.1 of~\cite{Sch86}).  Finally we place the reflex points of each channel.  The feasible region wherein each set of reflex points can be placed is bounded by lines through pairs of points already placed.  Thus, we can choose reflex points of polynomial size.  
\end{proof}

\ignore{
Now, given $G$, we can construct $T_1$ and $T_2$ in polynomial time, with all coordinates polynomial in the size of the original input. We can also compute $|E|$ in polynomial time. From Lemma~\ref{main-lemma}, $|\mbox{VC}(G)| \leq k$ if and only if $d(T_1, T_2)\leq 4k+28|E|$, which can be decided in polynomial time. This concludes our reduction.}

\section{Triangulations of point-sets}
We prove the NP-hardness of computing the flip distance between two triangulations of a point set by 
reducing from computing the flip distance between two  triangulations of a polygonal region. Given two triangulations $T_1$ and $T_2$ of a polygonal region $R$, we triangulate all the holes and pockets of $R$ the same way in both triangulations. 
Next, we repeat each edge on the boundary of the holes and pockets 
many times (as shown in Figure~\ref{fig:repeat-edges}) so that dismantling a boundary edge requires a large number of flips.  This gives two triangulations of a point set such that the flip distance between the two triangulations is the same as the flip distance between the original $T_1$ and $T_2$.
\begin{theorem}
\changed{The following problem is NP-complete:  Given two triangulations of a point set in the plane, and a number $k$, is the flip distance between the triangulations at most $k$?}
\end{theorem}

\begin{figure}
\centering
\includegraphics[scale=.8]{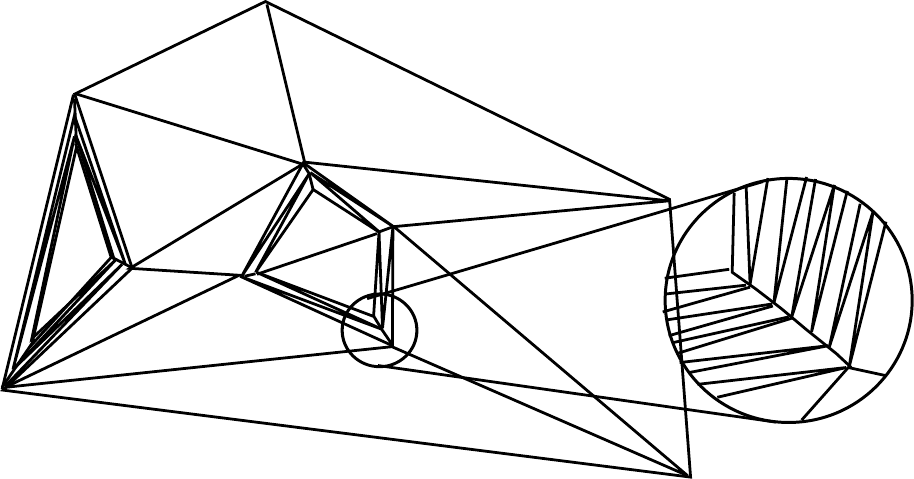}
\caption{Repeating edges on the boundary of holes.}
\label{fig:repeat-edges}
\end{figure}

\section{Conclusion}
We have shown that it is NP-complete to compute the flip distance for triangulations of a polygonal region, or  a point set. The problem remains open for a convex polygon, or a simple polygon, and also for more combinatorial objects such as labelled and unlabelled maximal planar graphs. 

\bigskip\noindent{\bf Acknowledgements.}
We thank Therese Biedl for helpful suggestions.

\bibliographystyle{abbrv}

\bibliography{flipping}

\newpage
\appendix
\section{Proof of Claim~\ref{claim:poly-size}}
As mentioned in Section~\ref{sec:details}, our proof consists of three steps:
\begin{enumerate}
\item Draw the channels with straight and parallel chains so that the feasible regions are non-empty and contained in a small square surrounding the original vertex. Moreover, the points $A, B, A', B', A'', B''$ get polynomial sized coordinates.
\item Obtain a point with polynomial sized coordinates inside each feasible region.
\item Make all the chains reflex.
\end{enumerate}

For the first part, we begin by placing a square of side $c$ around each vertex $v$ with $v$ at its center (Figure~\ref{fig:constraints}), such that $c$ is at least a constant factor smaller than the smallest edge in the drawing $\Gamma$ and can be written with polynomial number of bits. If an edge passes through a corner of the square, then we slightly increase one of its sides. Our aim is to find the points $A, B, A', B', A'', B''$ on the boundary of the square. (Note that even though the edge $AB$ will not be orthogonal to the two chains of the channel, the properties of the channel that we proved in Sections~\ref{sec:proof-idea} and \ref{sec:details} will still hold.)

The edges and their extensions (the dotted lines in Figure~\ref{fig:constraints}) intersect the square at points whose coordinates are polynomial size. Let $S$ be the set of intersection points and the corners of the square. For the edge corresponding to channel 1, consider the point $p_1$ where it intersects the square and find the point $p$ other than itself in $S$ that lies on the same edge of the square and is closest to it. Setting $A$ to be the point on the boundary of the square a distance $pp_1/3$ away from $p_1$ towards $p$ and $B$ the symmetric point on the opposite side determines the channel and its width. Do the same thing at the other end of the edge corresponding to channel 1 and obtain another width. Finally, pick the narrower of the two options for channel 1. Since $A$ and $B$ lie on the edge of the square and their distance to $p_1$ is polynomial, we need polynomial number of bits to express the coordinates of $A$ and $B$ as well. Repeat the above for $A', B', A''$ and $B''$. 

Since all the possible intersection points between the upper and lower chains of the channels occur inside the square, all the feasible regions have non-empty intersections with the interior of the square.

\begin{figure}
\centering
\subfigure[Degree-2]{
\includegraphics[scale=.85]{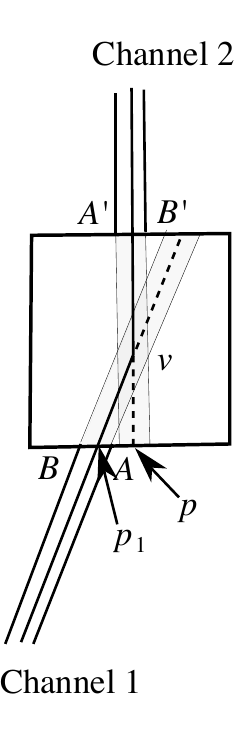}
\label{fig:constraints-two}
}
\subfigure[Degree-3]{
\includegraphics[scale=.85]{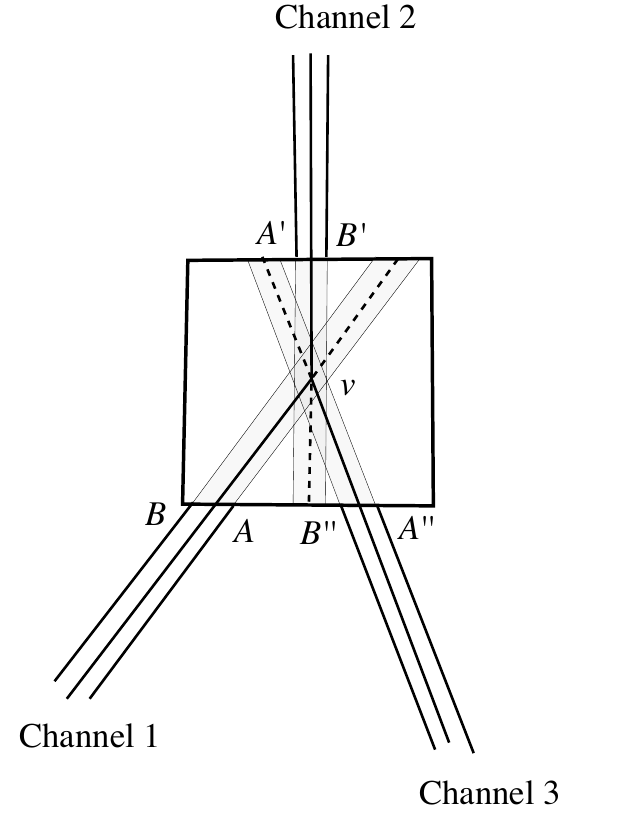}
\label{fig:constraints-three}
}
\caption{Constraints for vertex gadgets.}
\label{fig:constraints}
\end{figure}

Now, since the feasible regions are non-empty and are defined by linear inequalities with polynomial sized coefficients, using the theory of linear programming, we find a point with polynomial size coordinates inside each feasible region. Finally, to make the chains reflex, we find a location for each point on it one by one. Each point on the chain has a feasible region now defined by two kinds of constraints: 1) for every point outside the channel, if the point was inside the narrow mouth, it should remain inside and if it was outside the wide mouth, it should remain outside and 2) the new location of the point should maintain the reflexivity of the chain. These two constraints are also linear and thus we can find polynomial size coordinates for each point on the chains. This concludes our proof.

\end{document}